\documentclass[12pt,a4paper,reqno]{amsart}
\usepackage{amsmath}
\usepackage{geometry}
\usepackage{amssymb,latexsym}
\usepackage{verbatim}
\usepackage{extarrows}
\usepackage{enumerate}
\usepackage{txfonts}
\usepackage{mathtools}
\usepackage{bbm}
\newcommand{\1}{\mathbbm{1}}

\usepackage{mathtools}
\usepackage[tableposition=top]{caption}
\usepackage{booktabs,dcolumn}

\theoremstyle{plain}

\newtheorem{theorem}{Theorem}[section]
\newtheorem{proposition}[theorem]{Proposition}

\newtheorem{lemma}[theorem]{Lemma}

\theoremstyle{definition}

\newtheorem{definition}[theorem]{Definition}
\newtheorem{remark}[theorem]{Remark}

\newtheorem{example}[theorem]{Example}

\newcommand\R{\mathbb{R}}

\begin{document}

\title{On a relaxation of time-varying actuator placement}

 \author{Alex Olshevsky}
 \address{Department of Electrical and Computer Engineering and Division of System Engineering, Boston University}
 \email{alexols@bu.edu}

\begin{abstract} We consider the time-varying actuator placement in continuous time, where the goal is to maximize the trace of the controllability Grammian. A natural relaxation of the problem is to allow the binary $\{0,1\}$ variable indicating whether an actuator is used at a given time to take on values in the closed interval $[0,1]$. We show that all optimal solutions of both the original and the relaxed problems can be given via an explicit formula, and that, as long as the input matrix has no zero columns, the solutions sets of the original and relaxed problem coincide. 
\end{abstract}

\maketitle

\setcounter{tocdepth}{1}



\section{Introduction} \label{intro}

\subsection{Time-Varying Actuator Placement} We consider the time-varying actuator placement problem:  informally,  given a differential equation with input, we would like to optimize some controllability-related objective while using few nonzero inputs per time step. This is motivated by scenarios where setting an input to something nonzero at a given time carries a fixed cost that can be much larger than the cost of synthesizing the input itself. Our variation of the problem is ``time-varying,'' in the sense that we allow different inputs to be nonzero at different times; this is in contrast to ``fixed'' actuator placement problems, where one has to select the same set of actuators to be nonzero across all time.

Formally, we are given a differential equation with input 
\[ \dot{x}(t) = A x(t) + B u(t), ~~~~ 0 \leq t \leq T, \] where we assume that $A \in \R^{n \times n}$ and $B \in \R^{n \times m}$ is a matrix with no zero columns\footnote{If $B$ does have zero columns, then the corresponding entry of $u(t)$ does not affect $x(t)$. Consequently, we can simply delete the nonzero columns of $B$ and reindex the vector $u(t)$.}.  Our goal is to choose a diagonal matrix $V(t)$ whose entries lie in the binary set $\{0,1\}$ optimizing some controllability-related properties of the resulting differential equation
\[ \dot{x}(t) = A x(t) + B V(t) u(t).  \] The multiplication of the input $u(t)$ by the diagonal matrix $V(t)$ can be thought of as choosing to use only certain actuators. Indeed, if $V_{ii}(t)=0$, then $u_i(t)$ has no effect on $x(t)$, and the $i$'th entry of the input is ignored at time $t$. 

Typical controllability-related objectives are usually formulated in terms of the controllability Grammian, which we  define\footnote{It would be more standard to replace $t$ by $T-t$ in the definition of the controllability Grammian, but since that definition is equivalent to the one we give with a ``flipped'' $V(t)$, we prefer to avoid dealing with $T-t$'s throughout this paper.}
\[ W_V = \int_0^T e^{A t} B V(t) V(t)^T B^T e^{A^T t} ~ dt.\] The most natural objective is perhaps to minimize ${\rm Tr}(W_V^{-1})$, which is proportional to the average energy to move from the origin to a uniformly random point on the unit sphere (see e.g., discussion in  \cite{olshevsky2016eigenvalue}). However, this function is often challenging to reason about. For example, as a consequence of \cite{olshevsky2014minimal, tzoumas2015minimal} a number of optimization problems involving ${\rm Tr}(W_V^{-1})$  are NP-hard. 

We follow several recent papers which instead consider maximization of ${\rm Tr}(W_V)$. This is because ${\rm Tr}(W_V)$ is easier to reason about and can be used to construct bounds on ${\rm Tr}(W_V^{-1})$ (see discussion in \cite{nozari2017time,nozari2019heterogeneity, ikeda2018sparsity}). 
Furthermore, we will seek to do so in the presence of an upper bound on the number of actuators used per unit time step. More formally, denoting $V(t) = {\rm diag}(v(t))$, it is typically assumed that the diagonal entries $v_i(t) \in \{0,1\}$ satisfy the constraint
\[ \int_0^T \sum_{i=1}^m |v_i(t)| \leq \alpha, \] for some  $\alpha$.  We will refer to functions $v_i(t)$ satisfying these constraints as {\em feasible}. Note that, because we have constrained $v_i(t) \in \{0,1\}$, this is the same as requiring that 
\[ \sum_{i=1}^m \mu ( \{ t: v_i(t) = 1 \} ) \leq \alpha, \] where $\mu(\cdot)$ denotes the Lebesgue measure. We will naturally assume that  $\alpha \in (0,mT)$, as otherwise the problem is trivial.

A natural relaxation of the problem is to allow each $v_i(t)$ to lie in the closed interval $[0,1]$ instead of requiring it to take on the binary values $\{0,1\}$. We will refer to this as the {\em relaxed time-varying actuator placement problem}, and the version where $v_i(t)$ are required be in $\{0,1\}$ will be referred to as the {\em original time-varying actuator placement problem}.  These definitions lead to the main question which is the concern of this work, namely {\em understanding when the optimal solutions sets of the original and relaxed problem coincide.}

\subsection{Previous work} 
Our paper is most closely related to the recent work \cite{ikeda2018sparsity}, where the same question was considered. We next give a statement of the main results of \cite{ikeda2018sparsity}. 

Let us adopt the notation $b_j$ for the $j$'th column of the matrix $B$. Further, for $i=1, \ldots, m$, we consider the functions
\begin{equation} \label{eq:fdef} f_i(t) = b_i^T e^{A^T t} e^{A t} b_i. \end{equation} 
It is then possible to give a condition in terms of the functions $f_i(t)$ for the relationship between the original and relaxed time-varying actuator placement problems. 

\begin{theorem}[\cite{ikeda2018sparsity}] \label{thm:ikeda} The optimal solution set of the relaxed problem is non-empty. Further, assuming that assuming that  $f_i(t)$ is not constant for all $i \in \{1, \ldots, m\}$, we have that: 
\begin{enumerate} \item All optimal solutions of the relaxed problem take $\{0,1\}$ values almost everywhere. 
\item The optimal solution of the the solution sets of the original and relaxed actuator scheduling problems coincide. \end{enumerate} \end{theorem} 

This theorem is an amalgamation of Theorems 1-3 in \cite{ikeda2018sparsity}. It provides an answer to the motivating concern of the present paper. Related theorems in more general settings were also proved in \cite{ikeda2018sparse} and \cite{ikeda2019sparse}. However, the condition that $f_i(t)$ are not constant is only shown to be sufficient in this theorem. As can be seen (see Section \ref{sec:ex} below), this condition is not necessary for the two solution sets to coincide. 

Moreover, in discrete-time the optimal  schedule for the original problem can be found by a greedy method (see Theorem 5 of \cite{dilip2019controllability}). This suggests it may be possible to give a characterization of the optimal schedule for the original \& relaxed problems in the continuous-time model studied in this paper. 

Our paper is also related to the works \cite{nozari2017time, nozari2019heterogeneity} which studied combinatorial implications of maximization of the trace of the controllability Grammian, relating them to quantities like centrality and communicability in graphs. Also related is \cite{bof2016role} which studied combinatorial aspects of the smallest eigenvalue of the controllability Grammian, which is a measure of the maximum control energy to go from the origin to a point on the unit sphere.

Beyond that, the the time-varying actuator placement problem is quite old; for example, a version of it dates back to a paper of Athans in 1972 \cite{athans1972determination}. There is quite a bit of recent work on understanding efficient algorithms as well as fundamental limitations for this problem. For example, fundamental limitations in terms of unavoidably large control energy have been studied in \cite{pasqualetti2014controllability, olshevsky2016eigenvalue, zhao2016gramian, klickstein2018control} among others. Algorithms for actuator placement, in either the fixed or time-varying regime, based on randomized sampling \cite{bopardikar2017sensor, jadbabaie2018deterministic, siami2018deterministic, hashemi2018randomized}, convex relaxation \cite{zare2019proximal, summers2016convex}, or greedy methods \cite{zhang2017sensor, chanekar2017optimal, mackin2018submodular, silva2019model, zhao2016scheduling, summers2019performance} were studied in  recent works.   Given the relatively large amount of work done on different versions of the problem which are not directly related to our motivating concern, we refer the reader to the above papers for a broader overview of the field.

\subsection{Our contribution and the remainder of this paper} We show that, under our assumption that $B$ has no zero columns, the optimal solution sets of the original and relaxed problems always coincide. This comes out as a byproduct of an explicit formula for the solution of the relaxed problem. In turn, this  is done by drawing a connection to (a  modification) of the classical notion of a rearrangement of a function.

In Section \ref{sec:statement}, we give a statement of our main result and illustrate it with an example. The main result itself is proved in the  Section \ref{sec:proof}. Along the way, we will need to use many properties of the rearrangement of a function;  since these proofs are very similar to existing proofs in the literature for a slightly different notion of rearrangement, they are not included in the main text but relegated to the appendix for completeness. 

\section{Statement of the main result\label{sec:statement}}

\subsection{The (asymmetric) rearrangement} We need to introduce several concepts and notations to state our main result. 

We adopt the standard notation that the indicator function $1_{\mathcal{X}}(x)$ equals one if the point $x$ belongs to the set $\mathcal{X}$ and zero otherwise. Given a Lebesgue measurable subset $\mathcal{X} \subset \R$ of the real line of finite measure, we define its  rearrangement $\mathcal{X}^*$ to be the interval $[0,l]$ whose length $l$ is the same as the Lebesgue measure of $\mathcal{X}$. As already mentioned, we adopt the convention of using $\mu({\mathcal X})$ to denote the Lebesgue measure of the set ${\mathcal X}$.

Given a measurable nonnegative function $f: [0,a] \rightarrow \R$ with bounded range, its   rearrangement $f^*: [0,a] \rightarrow \R \cup \{ +\infty \}$ is defined as 
\begin{equation} \label{eq:rearr} f^*(x) =\int_0^{\infty} \mathbbm{1}_{\{ y \in [0,a] :  f(y) > t \}^*}(x) ~ dt \end{equation}

Intuitively, the rearrangement is that it corresponds to ``sorting'' the function $f(x)$. In particular:

\begin{proposition} The rearrangement $f^*(x)$ is measurable, non-increasing and its level sets have the same measure as the level sets of $f(x)$, i.e., for all $\alpha \in \R$,  \begin{eqnarray*} \mu ( \{ f(x) \geq \alpha \} ) & = & \mu(\{ f^*(x) \geq \alpha \}) \\ 
\mu ( \{ f(x) = \alpha \} ) & = & \mu(\{ f^*(x) =\alpha \}). 
\end{eqnarray*} \label{prop:volume}
\end{proposition}

We illustrate this with an example. 

\begin{example} Consider $f(x)=x^2$ defined on the domain $[0,1]$. In that case, $\{y \in [0,1]: f(y) > t\}$ is the set $\{ y \in [0,1]: y^2 > t\}$ which, for $t \in [0,1]$, equals $(\sqrt{t},1]$; when $t>1$, the set $\{y \in [0,1]: f(y) > t\}$ is empty. The rearrangement of $(\sqrt{t},1]$ is $[0,1-\sqrt{t}]$.
Thus, for any $x \in [0,1]$ we have that, \[ f^*(x) = \int_0^{1} \1_{[0,1-\sqrt{t}]}(x) ~dt = \int_0^{(1-x)^2} 1 ~dt = (1-x)^2. \] 
It is, of course, immediate that $f(x)=x^2$ and $f^*(x)=(1-x)^2$, when defined over the domain $[0,1]$, have level sets of the same measure. 
\end{example}

\subsection{Statement of the main result}  Our first step is to take the functions  $f_i:[0,T] \rightarrow \R, i =1, \ldots, m,$ defined in Eq. (\ref{eq:fdef}) and define  their ``concatenation'' $F(t): [0,mT] \rightarrow \R$  which consists on putting these functions ``side by side'' on the interval $[0,mT]$. Formally, 
\[ F(t) = f_i(t) \mbox{ when } t \in [(i-1)T, iT). \]

Since the functions $f_i(t)$ are clearly continuous, they have bounded range over $[0,T]$. Further, these functions are clearly nonnegative. As a result, $F(t)$ is also nonnegative and also has bounded range, and  therefore the rearrangement $F^*(t)$ is well-defined.

Our main result will be stated in terms of the function $F^*(t)$. To be rigorous, we also give a definition next of what it means for a function to be strictly decreasing from the left or the right; this definition is standard. 

\begin{definition} We will say that a function $g: [0,a] \rightarrow \R \cup \{+\infty\}$ is strictly decreasing to the right at $x_0 \in (0,a)$ if $g(x_0) > g(x_0 + \epsilon)$ for all $\epsilon$ small enough. Similarly, we will say that $g$ is strictly decreasing from the left at $x_0 \in (0,a)$ if $ g(x_0 - \epsilon) > g(x_0)$ for all $\epsilon$ small enough. 
\end{definition} 

The main result of this paper is the following theorem.

\bigskip


\begin{theorem} 

\begin{enumerate} \item Suppose $F^*(x)$ is strictly decreasing on the right at $x=\alpha$. Then the unique\footnote{Of course, we can modify any solution on a set of measure zero without any effect.  Thus, here and throughout the remainder of the paper, whenever we refer to equality of functions, we mean to up sets of zero measure.} optimal solution to the relaxed time-varying actuator placement problem is 
\[ v_i^{{\rm opt},{\rm r}}(t) = \1_{\{f_i(t) \geq F^*(\alpha)\}}. \]
\item Suppose $F^*(x)$ is strictly decreasing from the left at $x=\alpha$. Then the unique optimal solution to the relaxed time-varying actuator placement problem is 
\[ v_i^{\rm opt, l}(t) = \1_{\{f_i(t) > F^*(\alpha)\}}. \]
\item Suppose $F^*(x)$ is not strictly decreasing at $x=\alpha$ either from the left or from the right. Then the set of optimal solutions to the relaxed time-varying actuator placement problem has more than one element. However, all optimal solutions of the relaxed problem can be parametrized as  
\[ v_i^{\rm opt, nlr}(t) = \1_{\{ \tau: f_i(\tau) > F^*(\alpha) \}}(t)+ \1_{S_i}(t), \] for some sets $$S_i \subset \{ \tau :  f_i(\tau) = F^*(\alpha) \},$$ with $$\sum_{i=1}^m \mu(S_i) = \alpha - \mu(\{ \tau: F(\tau) > F^*(\alpha)\}).$$    
\end{enumerate} \label{thm:mainthm}
\end{theorem}

\bigskip

Since all the solutions exhibited in this theorem are binary, the immediate implication is that {\em the solutions of the original and relaxed problem are always the same} (though recall we assumed that $B$ has no zero columns). 

In particular, the condition that the functions $f_i(t)$ not be constant in Theorem \ref{thm:ikeda} is not necessary, once the pathological cases where $B$ has a zero column are ruled out. However, as we explain below in Remark \ref{remark1}, the non-constancy of $f_i(t)$ is nevertheless natural condition in this context, as it ensures that the solution of the relaxed problem is unique. 

The novelty of this theorem is two-fold. First, it shows that we can write down the optimal solution(s) of the relaxed problem via an explicit formula. Second, any desired properties of the optimal solution(s) can now be simply ``read off'' the formula. Moreover, once the notion of the rearrangement has been introduced, the theorem is quite intuitive: informally, it says that we have to take the ``top slice'' of the functions $f_1(t), \ldots, f_m(t)$. 

Finally, we mention that $F^*(t)$ can be strictly  decreasing from both the left and the right at $t=\alpha$, in which case both parts (1) and (2) of the theorem apply. There is no contradiction there, since, in this scenario, the two formulas for the solution differ only on a set of measure zero.


\begin{remark} A natural concern is whether the solutions above exhibit Zeno phenomena, i.e., whether they involve infinitely many switches in a finite interval. In the first two cases, when the function $F^*(t)$ is strictly decreasing from the left or right, this does not occur. This is because the functions $f_i(t)$ are analytic everywhere, and consequently the sets $\{ t: f_i(t) = F^*(\alpha) \}$ either equal $[0,T]$ or have finite cardinality. Thus one will either take $v_i(t)=1$ on all $[0,T]$, or $v_i(t)=0$ on all $[0,T]$, or $v_i(t)$ will switch finitely many times between $0$ and $1$ over the interval $[0,T]$. 

In the third case, the situation is slightly more complicated. In the event that some $f_i(t)$ is constant, the set $\{ t: f_i(t) = F^*(\alpha)\}$ may equal all of $[0,T]$, and an optimal solution might then correspond to choosing a subset of $[0,T]$. In that case, it is certainly possible to find an optimal solution that makes infinitely many switches in finite time. However, its also possible to pick the subset $S$ in the theorem statement to be an interval, avoiding this phenomenon.

To summarize, in the first two cases, Zeno phenomena do not occur; and in the third cases, optimal solutions avoiding Zeno phenomena always exist. 
\end{remark} 


\subsection{An example} \label{sec:ex}

 Consider the dynamical system 
\[ \dot{x} = 
\left( \begin{array}{ccc} 
0 & 0 & 0 \\
0 &  0 & 0 \\
0 & 0 & 1
\end{array} \right) x + 
\left( 
\begin{array}{ccc}
\gamma & 0 & 1 \\ 
0 & \gamma & 1 \\
0 & 0 & 1
\end{array}
\right) u \]

In this case, it is easy to compute the controllability Grammian explicitly. Indeed, 
\begin{eqnarray*} W_V & = & \int_0^T e^{A t} B V(t) V(t) B^T e^{A^T t} ~dt \\ 
& = & \int_0^T  v_1^2(t)\left( \begin{array}{ccc} 
\gamma^2  & 0 & 0 \\
0 &  0 & 0 \\
0 & 0 & 0
\end{array} \right) + 
v_2^2(t) \left( \begin{array}{ccc} 0 & 0 & 0 \\
0 & \gamma^2  & 0 \\
0 & 0 & 0
\end{array} \right) + 
v_3^2(t)\left( \begin{array}{ccc} 1 & 1 & e^t \\
1 & 1 & e^t \\
e^t & e^t & e^{2t}
\end{array} \right) ~~ dt
\end{eqnarray*} 
so that 
\begin{equation} \label{eq:trexp} {\rm Tr}(W_V) = \int_0^T v_1^2(t) \gamma^2 + v_2^2(t) \gamma^2 + v_3^2(t) (2 + e^{2t}) ~dt. \end{equation} Suppose we consider this problem over the interval $t \in [0,2]$ with the constraint 
\[ \sum_{i=1}^3 \int_0^T |v_i(t)| ~dt \leq 2. \] In other words, we stipulate that on average a single actuator per unit time is used.  
\begin{figure}
\includegraphics[width = 1.0\textwidth]{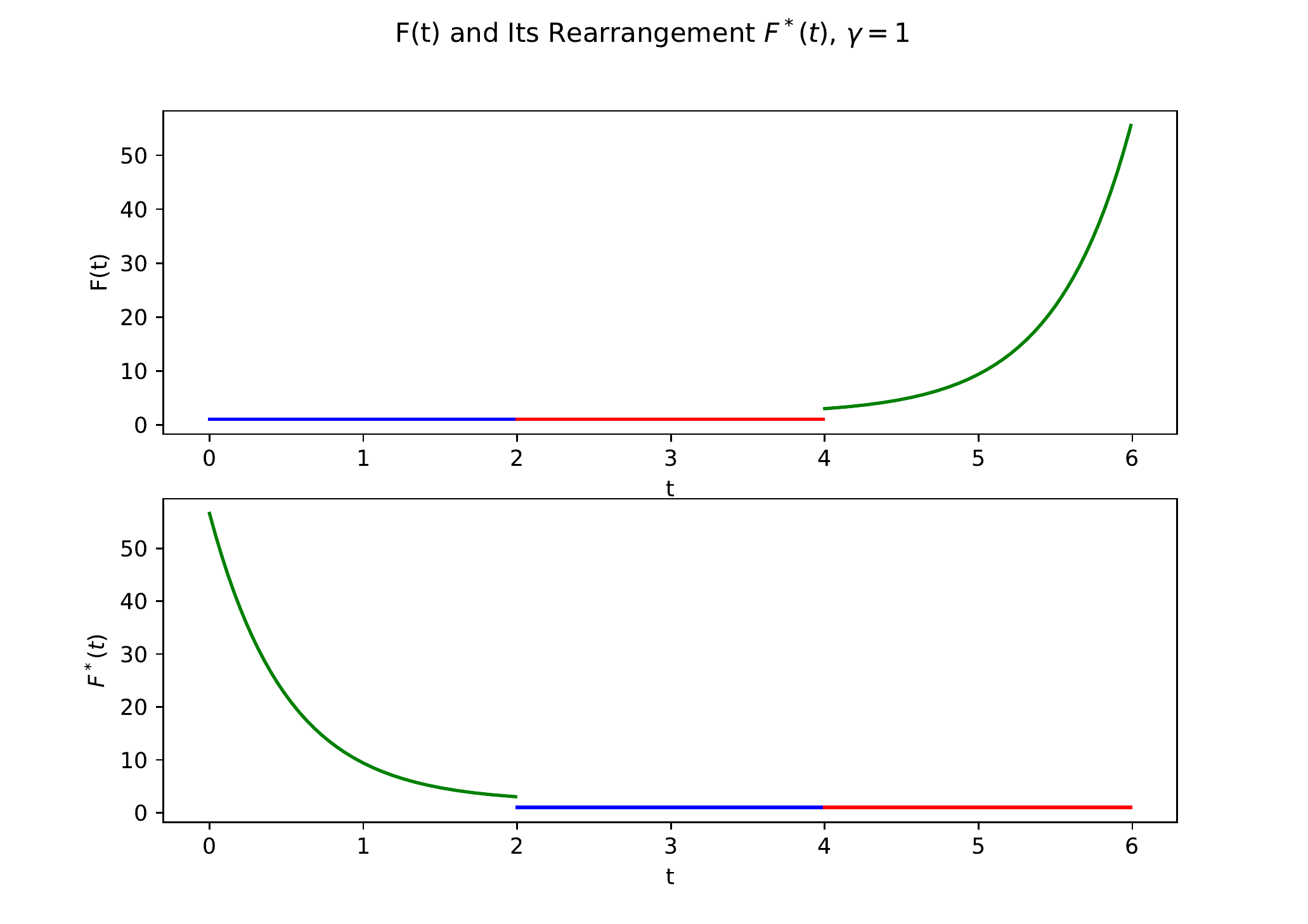}
\caption{The top graph shows the function $F(t)$ and the bottom graph shows the function $F^*(t)$, both when $\gamma=1$. The optimal actuator strategy is to take $v_i(t) = \1_{\{f_i(t) \geq F^*(2)\}}$ which selects only the third actuator.}
\end{figure}

\begin{figure}
\includegraphics[width = 1.0\textwidth]{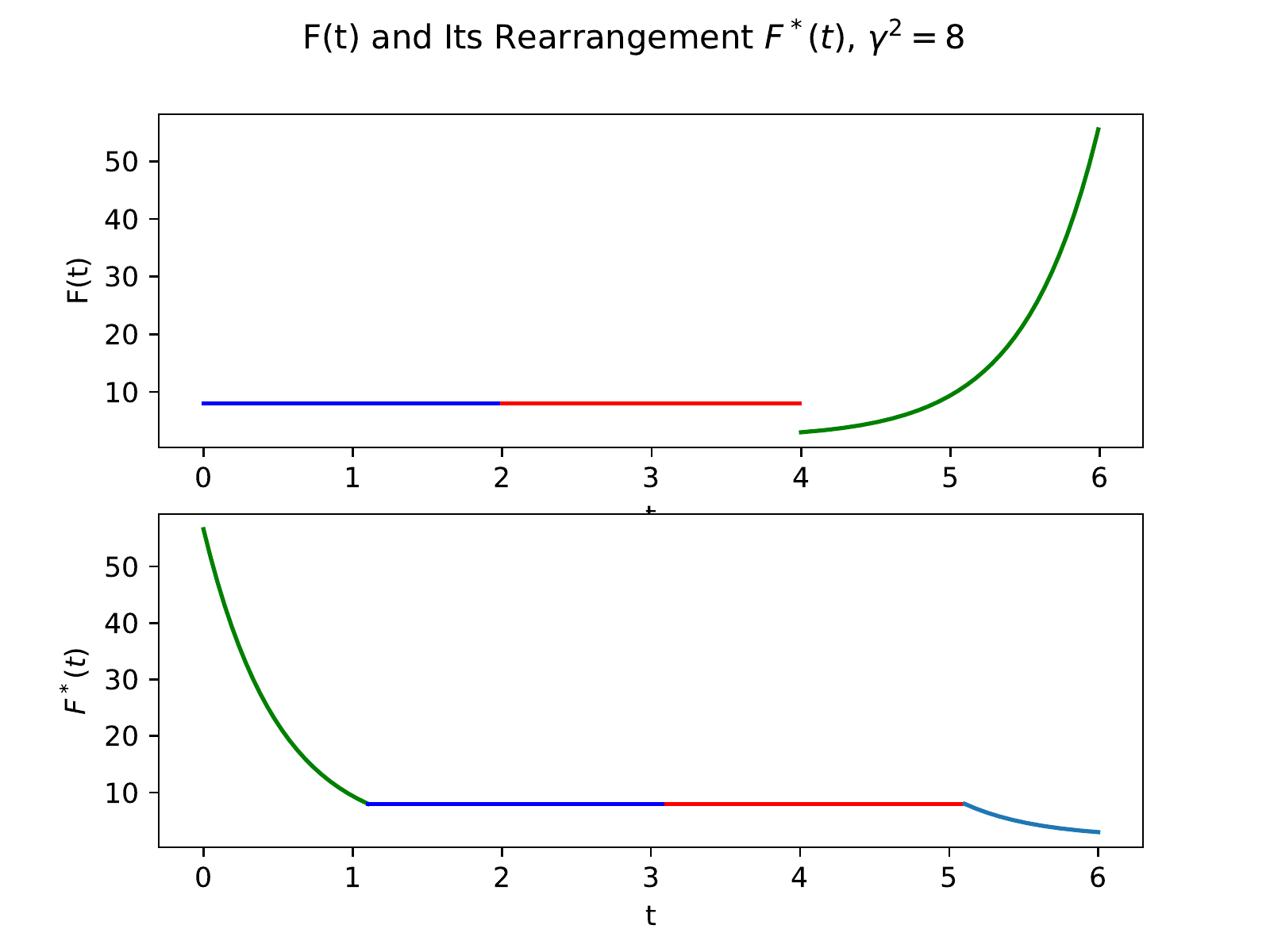}
\caption{The top graph shows the function $F(t)$ and the bottom graph shows the function $F^*(t)$ when $\gamma^2=8$. The optimal actuator strategy of taking $v_i(t) = \1_{\{f_i(t) \geq F^*(2)\}}$ forces us to select the third actuator for a period of time, and the remaining time can be split between the first and second actuator.}
\end{figure}

Let us consider several cases. We will make a number of assertions about the optimal solution; all of these follow immediately by inspection of Eq. (\ref{eq:trexp}), from which the optimal $V(t)$ can simply be read off. 

Let us consider first $\gamma=0$. In that case, we should set $v_1(t)=v_2(t)=0$, and set $v_3(t)=1$, for all $t \in [0,2]$. This optimal solution is, of course, unique. 

By contrast, suppose $\gamma^2 = 2 + e^{2}$. In that case,
the optimal solution will set $v_3(t)=1$ for $t \in [1,2]$, and we will have $v_1(t) = v_2(t)=0$ over the same interval. However, over $t \in [0,1]$, the optimal solution has $v_3(t)=0$. As for $v_1(t)$ and $v_2(t)$ over $t \in [0,1]$, they can be set to any functions with values in $\{0,1\}$ satisfying  
\[ \int_0^1 v_1(t) + v_2(t) = 1. \] In particular, we can decompose $[0,1] = S_1 \cup S_2$ for some disjoint $S_1,S_2$ and set $v_1(t) = \1_{S_1}(t), v_2(t) = \1_{S_2}(t)$. We see that the solution is not unique. 


Inspection of Eq. (\ref{eq:trexp}) reveals that the transition from having a unique solution to multiple solutions occurs when $\gamma^2 = 2$. This is exactly the point at which $F^*(t)$ goes from being strictly decreasing from the left at the point $t=2$ to being constant over an interval containing that point.  For an illustration, we refer the reader to Figures 1 and 2; the first figure shows a value of $\gamma$ for which $F^*(t)$ is strictly decreasing to the left at $t=2$, while the second figure shows how increasing $\gamma$ results in $F^*(t)$ which is flat over an interval containing $t=2$.

\section{Proof of the main result\label{sec:proof}}

In this section, we will prove our main result, Theorem \ref{thm:mainthm}. We begin by stating several properties of the rearrangement which we will find useful as propositions. 

 The  propositions below hold for all functions $f(x)$ and $g(x)$ such that their rearrangements can be defined, i.e., these functions must be nonnegative and have bounded range. Their proofs are fairly standard. Indeed, it is common in the literature to deal with the ``symmetric non-increasing rearrangement'' in which $f^*(x)$ is further constructed to be symmetric. For such a notion of rearrangement, the proofs of these facts appear in many places; a standard reference is Chapter 3 of the textbook \cite{lieb2001analysis}, which provides hints for many of these. Since our notion of rearrangement is a little different, we  provide the proofs of these propositions for completeness in the appendix. 
 
\begin{proposition}[Conservation of the $L^1$ norm] $$\int_0^a f(x) ~dx = \int_0^a f^*(x) ~dx $$ \label{prop:onenorm}
\end{proposition} 

\begin{proposition}[Hardy-Littlewood Inequality] $$\int_0^a f(x) g(x) ~dx \leq \int_0^a f^*(x) g^*(x) ~ dx$$ Moreover, we have equality if and only if for almost all $s,t$, we have that 
\[  \mu( \{ x: f(x) > t \} \cap \{ x: g(x) > s \}) = \min \left( \mu( \{ x: f(x) > t \}, \mu  (\{ x: g(x) > s \}) \right). \]
\label{prop:hl}
\end{proposition}

\begin{proposition}[Monotonicity] If $f(x) \leq g(x)$ for all $x$, then $f^*(x) \leq g^*(x)$ for all $x$. In particular, since a constant function is the rearrangement of itself, if $f(x) \leq 1$ for all $x \in [0,a]$, then $f^*(x) \leq 1$ for all $x \in [0,a]$.\label{prop:bound}
\end{proposition}

\begin{proposition}[Integral identity] Suppose  $b < a$. Then, if $f^*(x)$ is strictly decreasing on the right at $x=b$, we have that \label{prop:integral} 
\[ \int_0^a f(x) \1_{\{ f(x) \geq f^*(b) \} } ~dx = \int_0^{b} f^*(x) ~ dx
\] On the other hand, if $f^*(x)$ is strictly decreasing from the left at $x=b$, then 
\[ \int_0^a f(x) \1_{\{ f(x) > f^*(b) \} } ~dx = \int_0^{b} f^*(x) ~ dx
\]
\end{proposition}

\begin{proposition}[Integral identity] Let $(b_l,b_u)$ be the largest open interval containing $b$ on which  the function $f^*(x)$ is constant. Then if $S$ is a subset of the set $\{ x: f(x) = f^*(b)\}$, we have that for any $\delta \in \R$,
\[ \int_0^a f(x) \left( \1_{\{x: f(x) > f^*(b)\}} + \delta \1_{S} \right) ~dx
 = \delta  f^*(b)  \mu (S) + \int_0^{b_l} f^*(x) ~dx.
\] \label{prop:cint2}
\end{proposition}

With these propositions in place, we turn to our first lemma, which introduces some notation pertaining to functions $f^*(x)$ which are not strictly decreasing from either direction at a point.

\begin{lemma} Suppose $f: [0,a] \rightarrow \R$ and $f^*(x)$ is not decreasing from either the left or the right at the point $x=b$ with $b \in (0,a)$.  Then there exists an open interval $(b^l, b^r)$ containing $b$ such that
\begin{enumerate} \item $f^*(x)$ is constant on this interval.
\item $(b^l, b^r)$ is the largest open interval containing $b$ with this property. 
\item 
\begin{eqnarray*} 
\mu (\{ x: f(x) > f^*(b) \} ) & = & b^l \\ 
\mu ( \{ x: f(x) = f^*(b) \} ) & = & b^r - b^l \\
\mu ( \{ x: f(x) \geq f^*(b) \} ) & = & b^r
\end{eqnarray*}
\end{enumerate} \label{lem:lrc}
\end{lemma} 

\begin{proof} Since $f^*(x)$ is nonincreasing, it is immedate that if it is not strictly decreasing to the left or to the right at $x=b$, then it must be constant on some interval $(l,r)$ containing $b$. We can then define $a^l$ to be infimum of all $l$ such that $(l,r)$ is an interval containing $b$ on which $f^*(x)$ is constant; defining $b^r$ similarly, we obtain that $(b^l,b^r)$ is the largest open interval containing $b$ where $f^*(x)$ is constant. This proves parts (1) and (2). 

For part (3), we have that by Propdsition \ref{prop:volume},
\begin{eqnarray*} \mu (\{ x: f(x) > f^*(b) \} ) & = & \mu (\{ x: f^*(x) > f^*(b) \} ) \\
& = & \mu([0,b^l)) \mbox{ or } \mu([0,b^l]) \\ 
& = & b^l
\end{eqnarray*} and the proof of the second and third identities of item (3) proceed similarly. 

\end{proof}

Our next lemma is a straightforward generalization,  to the continuous space, of the fact that the largest convex combination of a set of numbers (subject to constraint on how big the weights can be) puts as much weight as possible on the largest  numbers. We present it without proof. 

\begin{lemma} Let $g(t): [0,a] \rightarrow \R$ be a nonincreasing function and suppose $b \in (0,a)$. Then,
\[ \max_{\gamma(t) \in [0,1], \int_0^a \gamma(t) ~dt = b} ~ \int_0^a  \gamma(t) g(t) ~dt \leq \int_0^b g(t) ~dt. \] Moreover,
\begin{enumerate} \item If $g(x)$ is either decreasing to the right or from the left at $x=b$, then the maximum is uniquely achieved by the function $\gamma(t)=\1_{[0,b]}(t)$.
\item If $g(x)$ is not decreasing from the right or the left at $t=b$, let $(b^l,b^r)$ be the open interval guaranteed by Lemma \ref{lem:lrc}. Then the functions which achieve the maximum are 
\[ \gamma^{\rm opt}(t) = \1_{[0,b^l)\}}(t) + \lambda(t) \1_{Q}(t), \]
where 
\[ Q \subset [b^l, b^r],\] and 
\[ \int_Q \lambda(t) ~dt = b - b^l. \]
\end{enumerate} \label{lem:convex}
\end{lemma}


We next exploit Lemma \ref{lem:convex} by applying it to the rearrangement of a function $f^*(x)$, which of course is nonincreasing. The result is stated as the following lemma. 

\begin{lemma} Let $a,b$ be scalars satisfying $b \leq a$ and  and let us define $\mathcal{B}$ as the set of nonnegative functions $\beta(t):[0,a] \rightarrow [0,1]$  with $$ \int_{0}^a \beta(t) ~dt \leq b.$$ Then 
\begin{equation} \label{eq:lem:ineq} \max_{\beta \in \mathcal{B}} \int_0^a \beta(t) f(t) ~ dt \leq \int_0^b f^*(t) ~dt. \end{equation} Moreover:
\begin{enumerate} \item If $f^*(t)$ is strictly decreasing on the right at $t=b$, then the unique $\beta(t)$ which achieves  this maximum is $\beta(t) = \1_{\{t: f(t) \geq f^*(b)\}}$. 
\item If $f^*(t)$ is strictly decreasing from the left at $t=b$, then the unique $\beta(t)$ which achieves this maximum is $\beta(t) = \1_{\{t: f(t) > f^*(b)\}}.$ 
\item If $f^*(t)$ is neither strictly decreasing from the left nor from the right at $t=b$, then the $\beta(t)$ which take values in $\{0,1\}$ almost everywhere which achieve this maximum are 
\[ \beta(t) = \1_{\{t: f(t) > f^*(b)\}} + \1_Q, \] where \[ Q \subset \{ t: f(t) = f^*(b)\} \] with \[ \mu(Q) = b - \mu ( \{ t: f(t) > f^*(b) \} ). \]
\end{enumerate} \label{lemma:ineq}
\end{lemma}

\begin{proof}  Observe we only need to prove Eq. (\ref{eq:lem:ineq}) with an inequality rather than equality, since by Proposition \ref{prop:integral} and Proposition \ref{prop:cint2}, the functions $\beta(t)$ given in items (1)-(3) make the left-hand side of Eq. (\ref{eq:lem:ineq}) equal to the right-hand side.

Indeed, using Proposition \ref{prop:hl} we have that for any  $\beta(t) \in \mathcal{B}$, 
\[ \int_0^a \beta(t) f(t) ~ dt \leq \int_0^{a} \beta^*(t) f^*(t) ~ dt. \] 
By Proposition \ref{prop:onenorm} we have that $\beta^*$ integrates to $b$ over $[0,a]$ just like $\beta$.  By Proposition \ref{prop:bound}, we have that $\beta^*(t) \leq 1$. Moreover, since the rearrangement of a function is always nonnegative as an immediate consequence of Eq. (\ref{eq:rearr}), we also have that $\beta^*(t)$ is nonnegative.  Thus: 
\[ \int_0^a \beta(t) f(t) ~ dt \leq \max_{\gamma(t) \in [0,1], ~\int_0^a \gamma(t) \leq b} ~~\int_0^a \gamma(t) f^*(t) ~dt. \] 
However, since $f^*(t)$ is nonincreasing by Proposition \ref{prop:bound}, by Lemma \ref{lem:convex} we must have that $\gamma^{\rm opt}(t) =  1_{[0,b]}(t)$. Thus 
\[ \int_0^a \beta(t) f(t) ~ dt \leq \int_0^{a} \1_{[0,b]}(t) f^*(t) ~ dt. \] Since this is true for any $\beta(t) \in \mathcal{B}$, this proves Eq. (\ref{eq:lem:ineq}).

It remains to analyze the case of equality. Suppose $f^*(t)$ is strictly decreasing from either the left or the right at $t=b$. Then by Lemma \ref{lem:convex} the optimal $\gamma(t)$ is unique, and therefore we must have $\beta^*(t) =  \1_{[0,b]}(t)$. In particular, this implies that $\beta(t)$ is the indicator function of an a set of Lebesgue measure $b$.  Let $S$ denote that set. 

From the equality conditions of Proposition \ref{prop:hl}, we need that for almost all $s,t$, 
\[ \mu (\{x: f(x) > t \} \cap \{ x: \1_S(x) \geq s \}) = \min ( \mu(\{x: f(x) > t \}, \mu(\{x: \1_S(x) \geq s\} ). \] This holds for allmost all $s,t$ if and only if it holds for allmost all $t$ and $s=1$, i.e., if \[ \mu ( \{ x : f(x) > t) \cap S ) = \min (\mu ( \{ x: f(x) > t \}), \mu(S)) \] for almost all $t$. For this to hold, we must have that for almost any level set of $\{ x: f(x) > t\}$,  up to a set of zero measure, $S$ is either contained in, or contains, that level set. Next, we analyze each case in the lemma statement separately. 

Fist suppose suppose $f^*(x)$ is strictly decreasing to the right at $x=b$. We will next argue that, up to a set of zero measure, we must have $S = \{ x: f(x) \geq f^*(b)\}$. Indeed, consider the level sets $\{x : f(x) \geq f^*(b)-\epsilon\}$. The measure of each of these sets is at least $b$. Since the measure of $S$ is exactly $b$, it follows that $S$ must be contained in almost all of these sets. Since these sets are nested, $S$ is in fact contained in all of them. So $S$ is contained in their intersection, which is $\{ x: f(x) \geq f^*(b)\}$. However, the latter set has measure $b$ by our assumption that $f^*(x)$ is decreasing to the right at $x=b$ (since $\mu (\{ x: f(x) \geq f^*(b)\}) = \mu( \{ x: f^*(x) \geq f^*(b)\})$ by Proposition \ref{prop:volume}, and the latter quantity equals $b$ since $f^*$ is decreasing to the right at $b$). Since $S$ also has measure $b$, it follows that $S$ actually equals $\{ x: f(x) \geq f^*(b)\}$ up to a set of zero measure. We have thus concluded the proof in the case when $f^*(x)$ is strictly decreasing to the right at $x=b$.

Next, suppose $f^*(x)$ is strictly decreasing from the left at $x=b$. We employ a nearly identical argument as the previous paragraph: we will argue that, up to a set of zero measure, we must have $S = \{ x: f(x) > f^*(b)\}$. Indeed, consider the level sets $\{x : f(x) \geq f^*(b)+\epsilon\}$. The measure of each of these sets is at most $b$. Consequently, $S$ contains almost all of them. Since these sets are nested, $S$ in fact contains all of them. So $S$ contains their union, which is $\{ x: f(x) > f^*(b)\}$. However, the latter set has measure $b$ by our assumption that $f^*(x)$ is decreasing from the left  at $x=b$. Since $S$ also has measure $b$, it follows that $S$ actually equals $\{ x: f(x) > f^*(b)\}$ up to a set of zero measure.  

Finally, suppose $f^*(x)$ is neither strictly increasing from the left nor from the right.  Lemma \ref{lem:lrc} ensures the  existence of the largest open $(b^l,b^r)$ containing $b$ on which $f^*(x)$ is constant. By Lemma \ref{lem:convex}, the set of optimal  $\gamma(t)$  equals
\[ \gamma^{\rm opt}(t) = \1_{[0,b^l)}(t) + \lambda(t) 1_{Q'}(t), \] for some subset $Q' \subset [b^l, b^r]$ with 
\begin{eqnarray*} \int_{Q'} \lambda(t) ~dt & = & b - b^l \\ 
& = & b -  \mu ( \{ x: f(x) > f^*(b) \} ),
\end{eqnarray*} with the second equality following by Lemma \ref{lem:lrc}. 

Since, in this part of the lemma, we only discuss  $\beta(t)$ taking on values in $\{0,1\}$, we must have that $\beta(t)$ is a binary function. We thus have that 
\[ \beta^*(t) = \1_{[0,b^l)}(t) + 1_{Q''}(t) \] for some subset $Q'' \subset [b^l, b^r]$ with 
\[ \mu(Q'') = b-b^l = \mu ( \{ x: f(x) > f^*(b) \} ). \] In particular, this means that $\beta(t)$ must be the indicator function of a set $S$ of measure $b = b^l + (b-b^l)$.

The rest of the proof proceeds similarly. We obtain, just as above, that for almost all level sets $\{ x: f(x) > c \}$, the set $S$ must contain, or be contained, in that level set. We then argue that since for $\epsilon>0$, \begin{eqnarray*} 
\mu ( \{ x: f(x) > f^*(b) - \epsilon \} )
& \geq & \mu ( \{ x: f(x) \geq f^*(b) \} ) \\ 
& = & b^r \\ 
& > & b
\end{eqnarray*} we have that that $S$ must be contained in  the sets 
$\{ x: f(x) > f^*(b) - \epsilon \}$ for almost all $\epsilon$. Since these sets are nested, $S$ is in fact contained in all of them, and, finally, it is contained in their intersection. Consequently, we have 
\[ S \subset \{ x: f(x) \geq f^*(b) \}. \] 

A similar argument gives that $\{ x: f(x) > f^*(b) \} \subset S$. Indeed, Lemma \ref{lem:lrc} tells us that for any $\epsilon > 0$, 
\begin{eqnarray*} 
\mu ( \{ x: f(x) > f^*(b) + \epsilon \} )
& \leq & \mu ( \{ x: f(x) > f^*(b) \} ) \\ 
& = & b^l \\ 
& < & b
\end{eqnarray*}  so that $\{ x: f(x) > f^*(b) + \epsilon\} \subset S$ for almost all $\epsilon$. Because these sets are nested, in fact all of them are contained in $S$; and consequently their union, which is $\{ x: f(x) > f^*(b) \}$ is contained in $S$. 

In summary, we have shown that 
\[ S = \{ x: f(x) > f^*(b) \} \cup \mbox{ some subset of } \{ x: f(x) = f^*(b) \} 
\] Let $Q$ be the ``some subset'' in the above equation.
We have that 
\[ b = \mu(S) = b^l + \mu(Q),\] 
which finally gives that 
\[ \mu(Q) = b-b^l = b - \mu ( \{ x: f(x) > f^*(b) \}),\] and the proof is finished.

\end{proof} 

With this last lemma in place, we can now turn to the proof of our main result.

\begin{proof}[Proof of Theorem \ref{thm:mainthm}] We first argue that the solutions we have proposed are feasible. First, uppose that $F^*(x)$ is strictly decreasing from the right at $x=\alpha$. Then, we have
\begin{eqnarray*} \sum_{i=1}^m \int_0^T |v_i^{\rm opt, r}(t)| & = & \sum_{i=1}^m \mu(\{ f_i(t) \geq F^*(\alpha)\}) \\
& = & \mu (\{F(t) \geq F^*(\alpha)\}) \\ 
& = & \mu (\{F^*(t) \geq F^*(\alpha)\}) \\ 
& = & \alpha,
\end{eqnarray*} where the last step follows from the assumption that $F^*$ is decreasing to the right at $\alpha$. 

Similarly, suppose $F^*(x)$ is strictly decreasing from the left at $x=\alpha$. Then
\begin{eqnarray*} \sum_{i=1}^m \int_0^T |v_i^{\rm opt, l}(t)| & = & \sum_{i=1}^m \mu(\{ f_i(t) > F^*(\alpha)\}) \\
& = & \mu (\{F^*(t) > F^*(\alpha)\}) \\ 
& = & \alpha,
\end{eqnarray*} where the last step follows from the assumption that $F^*$ is strictly decreasing from the left at $\alpha$. 

Finally, suppose $F^*(x)$ is not strictly decreasing from either the left or the right at $x=\alpha$. Then 

\begin{eqnarray*} \sum_{i=1}^m \int_0^T |v_i^{\rm opt, nlr}(t)| & = & \sum_{i=1}^m \mu(\{ f_i(t) > F^*(\alpha)\}) + \mu(S) \\
& = & \mu (\{F^*(t) > F^*(\alpha)\}) + (\alpha - \mu (\{F^*(t) > F^*(\alpha)\})) \\ 
& = & \alpha,
\end{eqnarray*}

Since it is immediate that all $v_i^{\rm opt, r}(t), v_i^{\rm opt, l}(t), v_i^{\rm opt, nlr} \in [0,1]$, we conclude that in all cases the proposed solutions are feasible. 

Let us adopt the notation $J_{v^{\rm opt, r}}$ for the cost corresponding to the functions $v_i^{\rm opt, r}$.  Let us compute this cost under the assumption that $F^*(t)$ is decreasing from the right at $t=\alpha$. We have that
\begin{eqnarray}  J_{v^{\rm opt, r}} & = &  {\rm Tr} \int_0^T \sum_{i=1}^m e^{A t} b_i  \1_{\{f_i(t) \geq F^*(\alpha)\}} b_i^T e^{A^T t} ~ dt \nonumber \\ 
& = &  \int_0^T \sum_{i=1}^m   \1_{\{f_i(t) \geq F^*(\alpha)\}} b_i^T e^{A^T t} e^{A t} b_i  ~ dt \nonumber \\ 
& = &  \int_0^T \sum_{i=1}^m  \1_{\{f_i(t) \geq F^*(\alpha)\}} f_i(t) ~ dt \nonumber \\ 
& = & \int_{0}^{mT} F(t) \1_{\{F(t) \geq F^*(\alpha)\}} ~dt \nonumber \\
& = & \int_0^{\alpha} F^*(t) ~dt,\label{eq:optcost} \end{eqnarray}  and the last step used Proposition \ref{prop:integral} and the assumption that $F^*(t)$ is strictly decreasing from the right at $t=\alpha$. Thus we have shown that the  choice of functions $v_i^{\rm opt, r}(t)$ achieves a cost of $\int_0^{\alpha} F^*(t) ~dt$.

We next argue that, under the assumption that $F^*(t)$ is strictly decreasing from the left at $t=\alpha$, the functions $v_i^{\rm opt, l}(t)$ achieve the same cost via a nearly identical argument:
\begin{eqnarray}  J_{v^{\rm opt, l}} & = &  {\rm Tr} \int_0^T \sum_{i=1}^m e^{A t} e_i  \1_{\{f_i(t) > F^*(\alpha)\}} e_i^T e^{A^T t} ~ dt \nonumber \\ 
& = &  \int_0^T \sum_{i=1}^m  \1_{\{f_i(t) > F^*(\alpha)\}} f_i(t) ~ dt \nonumber \\ 
& = & \int_{0}^{mT} F(t) \1_{\{F(t) > F^*(\alpha)\}} ~dt \nonumber \\
& = & \int_0^{\alpha} F^*(t) ~dt \end{eqnarray}  and the last step used Proposition \ref{prop:integral} and the assumption that $F^*(t)$ is strictly decreasing from the left at $t=\alpha$.

We finally argue that, under the assumption that $F^*(t)$ is decreasing neither from the left nor the right, the functions $v_i^{\rm opt, nlr}$ achieve the same cost. Indeed, let $(\alpha^l, \alpha^r)$ be the largest open interval containing $\alpha$ on which $F^*(t)$ is non-decreasing whose existence was guaranteed by Lemma \ref{lem:lrc}. That lemma tells us that $
\mu \{ t: F^*(t) > F^*(\alpha)\}  =  \alpha^l$.

Define $S_i' = S_i + (i-1)T$, where $S_i$ comes from the theorem statament and we translate it by $(i-1)T$ to make it so that $S_i' \subset [(i-1)T, iT]$. Further defining $S = \cup_{i=1}^m S_i'$ and proceeding similarly as before,  
\begin{eqnarray*} 
J_{v^{\rm opt, nlr}} & = & \int_{0}^{mT} F(t) \left( \1_{\{ F(t) > F^*(\alpha\}} + \1_{S} \right)\\
& = & \int_0^{\alpha^l} F^*(t) ~dt +  \mu(S) F^*(\alpha) \\ 
& = & \int_0^{\alpha^l} F^*(t) ~dt + (\alpha - \mu(\{F(t) > F^*(\alpha)\})) F^*(\alpha) \\ 
& = & \int_0^{\alpha^l} F^*(t) ~dt + (\alpha - \alpha^l ) F^*(\alpha) \\ 
& = & \int_0^{\alpha} F^*(t) ~dt,
\end{eqnarray*} where we relied on Proposition \ref{prop:cint2} in the second step.

To summarize, we have shown that under the appropriate assumptions, all three of $v_{i}^{\rm opt, r}(t)$ and $v_i^{\rm opt, l}(t)$ and $v_i^{\rm opt, nlr}(t)$ achieve the cost of 
\[ \int_0^{\alpha} F^*(t) ~dt. \]

We next show that, for any choice of functions $v_i(t)$, $i=1, \ldots, m$, $t \in [0,T]$, we will attain a cost that is upper bounded by $\int_0^{\alpha} F^*(t) ~dt.$. This part of the argument does not use assume anything about the behavior of $F^*(t)$ at $t=\alpha$. Let us adopt the notation $J_v$ for the cost corresponding to the functions $v_i(t), i = 1, \ldots, m$; our goal is to show that, as long as $v_i(t)$ are feasible, we have $J_v \leq \int_0^{\alpha} F^*(t)$. 

Indeed, for any feasible functions $v_i(t)$ we have that $v_i(t) \in [0,1]$ which means that
\begin{eqnarray} J_v & = & {\rm Tr} \int_0^T \sum_{i=1}^m v_i^2(t) e^{A t} b_i b_i^T e^{A^T t} ~dt \label{eq:first} \\ 
& \leq & {\rm Tr} \int_0^T \sum_{i=1}^m v_i(t) e^{A t} b_i b_i^T  e^{A^T t} ~dt ~~~~~~ \label{eq:square} \\ 
& = & \sum_{i=1}^m \int_0^T v_i(t) ~ {\rm Tr} \left(  e^{A t}  b_i b_i^T  e^{A^T t} \right) ~dt  \nonumber \\
& = & \sum_{i=1}^m \int_0^{T} v_i(t) f_i(t) ~dt \\
& = & \int_0^{mT} q(t) F(t) ~dt \label{eq:geq}
\end{eqnarray} where we define
\[ q(t) = \sum_{i=1}^m v_i(t - (i-1)T) 1_{t \in [(i-1)T, iT]}
\] 
Observing that 
\[ \int_0^{mT} |q(t)| =\sum_{i=1}^m \int_0^{T} |v_i(t)| \leq \alpha,\] by feasibility of $v_i(t)$, we can now apply  Lemma \ref{lemma:ineq} to Eq. (\ref{eq:geq}) and  obtain
\begin{eqnarray} 
J_v  
& \leq & \int_0^{\alpha} F^*(t) ~dt  \label{eq:last}
\end{eqnarray} where the final line followed by Eq. (\ref{eq:optcost}).

Putting it all together, we have thus shown that, under the appropriate assumptions, each of $v_i^{\rm opt, l}(t), v_i^{\rm opt,r}(t), v_i^{\rm opt, nlr}(t)$ are optimal. It remains to prove that these are the only optimal choices. For this, we must analyze the cases of equality in the above bounds. 

Observe that to achieve equality  we need to have equality
starting from the Equation of Eq. (\ref{eq:first}) through to Eq. (\ref{eq:last}). In particular, we must have equality in the in the application of Lemma  \ref{lemma:ineq}. But that lemma spells out conditions for equality. In particular, Lemma \ref{lemma:ineq}  forces $q(t)=\1_{\{t: F(t) \geq F^*(\alpha)\}}$ when $F^*(t)$ is decreasing from the right at $\alpha$ and $q(t)=\1_{\{t: F(t) > F^*(\alpha)\}}$ when it is decreasing from the left. By definition of $q(t)$, this is the same as having $v_i(t) = \1_{\{t: f_i(t) \geq F^*(\alpha)\}}$ and $v_i(t) = \1_{\{t: f_i(t) \geq F^*(\alpha)\}}$.  This concludes the proof for the case where $F^*(t)$ is either strictly decreasing to the right or from the left at $t=\alpha$. 

It only remains to analyze the cases of equality in the case where $F^*(t)$ is not strictly decreasing from either the left or the right at $t=\alpha$. First, observe that because $B$ has no zero columns and $e^{At}$ is always nonsingular, we have that 
\[ {\rm Tr}(e^{A t} b_i b_i^T e^{A^T t}) > 0,\] for all $i \in \{1, \ldots, m\}$ and $t \in [0,T]$. In particular, because $v_i(t) \in [0,1]$, the implication of this is that  to achieve equality going from Eq. (\ref{eq:first}) to Eq. (\ref{eq:square}), we must have that each $v_i(t) \in \{0,1\}$ almost everywhere. This implies the function $q(t)$ must be binary as well. 

Having established that, we apply the conditions for equality in the last item of Lemma \ref{lemma:ineq}. That lemma tells us that we must have 
\[ q(t) = \1_{\{t: F(t) > F^*(\alpha)\}} +  \1_{S},\] for some subset $S$ of the set $\{ t: F(t) = F^*(\alpha)\}$ with 
\[ \mu(S) = \alpha - \mu( \{ t: F^*(t) > F^*(\alpha)\}). \] This concludes the proof.

\end{proof}

\begin{remark} \label{remark1} Observe that   if all the $f_i(t)$ are not constant, then the solution of the relaxed actuator scheduling problem is unique. 
 Indeed, because each $f_i(t)$ is analytic everywhere and non-constant, the set $\{t \in [0,T]: f_i(t) = c\}$ is finite for all $c$, and consequently has Lebesgue measure zero. In particular, $\mu(\{t: F^*(t) = c\}) = 0$. This implies that $F^*(t)$ cannot be constant on any interval. By Lemma \ref{lem:lrc}, we conclude that $F^*(t)$ must either be strictly decreasing from the left or to the right at any point in $(0,mT)$. Theorem \ref{thm:mainthm} now implies that the solution of the relaxed actuator scheduling problem is unique. 

\end{remark}

\bibliographystyle{plain}
\bibliography{refs}

\section{Appendix}

We give the proofs of the propositions claimed in the body of the paper in this appendix. We stress again that these propositions are standard, and  for the standard notion of symmetric nonincreasing rearrangement they are widely used (e.g., see Chapter 3 of \cite{lieb2001analysis} which provides proof hints for many of them). For purposes of completeness, we spell out these proofs explicitly here.

\begin{proof}[Proof of Proposition \ref{prop:volume}]  We first argue that $f^*(x)$ is non-increasing. From the definition of rearrangement, this follows if we can show that the functions 
$\mathbbm{1}_{\{ y \in [0,a] :  f(y) > t \}^*}(x)$ are non-increasing in $x$ for any fixed $t$. In other words, we need to show that if $0 \leq x_1  \leq x_2 \leq a$ and 
\[ x_2 \in \{ y \in [0,a] :  f(y) > t \}^* \] then we have that
\[ x_1 \in \{ y \in [0,a] :  f(y) > t \}^*. \] But this immediately follows because the set $ \{ y \in [0,a] :  f(y) > t \}^*$ is an interval of the form $[0,l]$, so that if some $x_2$ belongs to it, then a smaller (but nonnegative) $x_1$ belongs to it as well. 

The remainder of Proposition \ref{prop:volume} will be proved by arguing that the level sets of $f(x)$ and $f^*(x)$ are almost rearrangements of each other, i.e., we will prove that the two sets  $\{ x: f(x) > u \}^*$ and $\{ x : f^*(x) > u \}$ differ in at most one element. 

Indeed, 
suppose $x' \notin \{ x: f(x) > u \}^*$. Then the supremum over all $t$ such that $x'$ belongs $\{ x: f(x) > t \}^*$ is upper bounded by $u$. Then,
\[ f^*(x') = \int_0^{+\infty} \1_{\{y \in [0,a] : f(y) > t\}}^*(x') ~dt \leq \int_0^u 1 ~dt \leq u.\] In particular, we conclude that $x' \notin \{ x : f^*(x) > u \}$.

On the other hand, suppose $x' \in \{ x: f(x) > u \}^*$. Since 
\[ 
\{ x: f(x) > u  \} = \cup_{\epsilon > 0} \{ x: f(x) > u + \epsilon \}, \] and the sets on the right hand side are nested, it follows that 
\[ \mu ( \{ x: f(x) > u \}  ) = \lim_{\epsilon \rightarrow 0} \mu (\{ x: f(x) > u + \epsilon \}). \] Now the set $\{ x: f(x) > u  \}^*$ is, by definition, an interval. A consequence of the above observation is that if $x'$ is not the right-endpoint of that interval, then 
there exists a positive  $\epsilon > 0$ such that 
$x' \in \{ x: f(x) > u+\epsilon \}^*$. It follows that $x'$ also belongs to all sets of the form $\{ x: f(x) > u_0 \}^*$ with $u_0 \leq u + \epsilon$. Therefore:
\[ f^*(x') = \int_0^{+\infty} \1_{\{y \in [0,a] : f(y) > t\}}^*(x') ~dt  \geq \int_0^{u+\epsilon} 1 ~dt = u + \epsilon,\] so that $x' \in \{ x : f^*(x) > u \}$.  To summarize, we have shown that $\{ x: f(x) > u \}^*$ and $\{ x : f^*(x) > u \}$ differ in at most one element (i.e., the right endpoint of the interval $\{ x: f(x) > u \}^*$). This proves that these two sets have the same measure.

Finally, we observe that 
 
\begin{eqnarray*} \mu (\{ x: f(x) = c\}) & =&  \lim_{\epsilon \rightarrow 0} \mu \{ x : f(x) > c+\epsilon\} - \mu \{x: f(x) > c-\epsilon\} \\ 
& = & \lim_{\epsilon \rightarrow 0} \mu \{ x : f^*(x) > c+\epsilon\} - \mu \{x: f^*(x) > c-\epsilon\} \\ 
& = & \mu (\{ x: f^*(x)=c) \}.
\end{eqnarray*} 
This completes the proof that level sets have the same measure. 

We conclude the proof by observing that $f^*(x)$ is measurable as a consequence of the fact that, as a consequence of the proof above,  $\{ x: f^*(x) > t \}$ is either a closed or a half-open interval.

\end{proof} 

The proof of the following propositions will require the so-called  ``layer cake representation,'' which states that if $g(x)$ is a nonnegative function, then 
\begin{eqnarray*} g(x) & = &  \int_0^{+\infty} \1_{\{ y: g(x) > y\}}(t) ~dt. \\
& = & 
\int_0^{+\infty} \1_{\{ y: g(x) \geq y\}}(t) ~dt.
\end{eqnarray*}

\begin{proof}[Proof of Proposition \ref{prop:onenorm}] 
Indeed,
\begin{eqnarray*} 
\int_0^a f(x) ~dx & = & \int_0^a \int_0^{+\infty} \1_{\{ y: f(x) > y\}}(t) ~dt ~dx \\ 
& = & \int_0^a \int_0^{+\infty}  \mu(\{ t: f(x) > t\}) ~dt ~dx \\ 
& = & \int_0^a \int_0^{+\infty}  \mu(\{ t: f^*(x) > t\}) ~dt ~dx \\ 
& = & \int_0^a \int_0^{+\infty} \1_{\{ y: f^*(x) > y\}}(t) ~dt ~dx \\ 
& = & \int_0^a f^*(x) ~ dx
\end{eqnarray*} where the third equality used Proposition \ref{prop:volume}. Finally, since both $f(x)$ and $f^*(x)$ are nonnegative by assumption, we can put the absolute values on the first and last integrand.
\end{proof} 

\begin{proof}[Proof of Proposition \ref{prop:hl}] Using the layer-cake representation,
\begin{eqnarray*}
\int_0^a f(x) g(x) ~ dx & = & \int_0^a \int_0^{+\infty} \int_0^{+\infty} \1_{\{y: f(x) > y\}}(t) \1_{\{y: g(x) > y\}}(s) ~ ds ~dt ~dx \\ 
& = &  \int_0^{+\infty} \int_0^{+\infty} \int_0^a \1_{\{y: f(x) > y\}}(t) \1_{\{y: g(x) > y\}}(s) ~ dx ~dt ~ds \\ 
& = & \int_0^{+\infty} \int_0^{+\infty} \mu( \{ x: f(x) > t \} \cap \{ x: g(x) > s \})  ~dt ~ds \\
& \leq & \int_0^{+\infty} \int_0^{+\infty}   \min \left( \mu( \{ x: f(x) > t \}, \mu  (\{ x: g(x) > s \}) \right)  ~dt ~ds,
\end{eqnarray*} since, in general,
\[ \mu (\mathcal{X} \cap \mathcal{Y}) \leq \min (\mu(\mathcal{X}), \mu(\mathcal{Y})). \] Now using Proposition \ref{prop:volume}, we have that
\begin{eqnarray*} \int_0^a f(x) g(x) ~ dx & \leq & \int_0^{+\infty} \int_0^{+\infty}  \min ( \mu( \{ x: f(x) > t \}^*, \mu  (\{ x: g(x) > s \}^*) ) ~dt ~ds \\
& = & \int_0^{+\infty} \int_0^{+\infty}  \mu ( \{ x:  f(x) > t \}^* \cap \{ x:  g(x) > s \}^* )  ~dt ~ds,
\end{eqnarray*} where the last step uses that the sets $\{x:  f(x) > t \}^*$ and $\{ x: g(x) > s \}^*$ are, by definition, both intervals of the form $[0,l]$. 

Using that the two  $\{ x: f(x) > u \}^*$ and $\{ x : f^*(x) > u \}$ differ in at most one element (as proven above), we move  to the representation in terms of indicator functions, 
\begin{eqnarray*} \int_0^a f(x) g(x) ~ dx & \leq &
\int_0^{+\infty} \int_0^{+\infty} \int_0^a \1_{\{y: f^*(x) > y\}}(t) \1_{\{y: g^*(x) > y\}}(s) ~ dx ~dt ~ds \\
& = & \int_0^a \int_0^{+\infty} \int_0^{+\infty}  ~ \1_{\{y : f^*(x) > y\}}(t) \1_{\{y: g^*(x) > y\}}(s) ~dt ~ds ~dx \\ 
& = & \int_0^a f^*(x) g^*(x) ~ dx
\end{eqnarray*} The inequality is proved, and it remains only to analyze conditions for equality. Inspecting the above proof, there is only one place where an inequality appears (namely in the fourth line of the proof). The condition for equality simply gives a necessary and sufficient condition for that line to be an equality.

Equivalently, we can use the version of the layer-cake representation with a non-strict inequality to obtain the same condition for equality but with nonscrict inequalities.

\end{proof} 

\begin{proof}[Proof of Proposition \ref{prop:bound}] We first show that the rearrangement of a constant is itself. Suppose $f(x)=c$ for all $x \in [0,a]$. Then, for all $x \in [0,a]$, 
\begin{eqnarray*} f^*(x) & =&  \int_0^{+\infty} \1_{\{ y \in [0,a]: f(y) > t \}^* }(x) ~dt \\
& =&  \int_0^{+\infty} \1_{\{ y \in [0,a]: c > t \}^* }(x) ~dt \\
& = & \int_0^{c} \1_{[0,a]^*}(x) ~ dt \\ 
& = & \int_0^{c} \1_{[0,a]}(x) ~dt \\ 
& = & c
\end{eqnarray*}
Next, if $f(x) \leq g(x)$ for all $x \in [0,a]$, then we have that for all $t$,
\[ \{ y : f(y) > t \} \subset \{ y: g(y) > t \},\]  so that for all $x,t$,
\[ \1_{\{ y: f(y) > t \} }(x) \leq \1_{\{ y: g(y) > t \} }(x),\] which, by definition of rearrangement, implies $f^*(x) \leq g^*(x)$.

\end{proof} 

\begin{proof}[Proof of Proposition \ref{prop:integral}] Suppose $f^*(x)$ is strictly decreasing to the right at $x=b$. Starting with the layer cake decomposition, 
\begin{eqnarray*} 
\int_0^a f(x) \1_{\{f(x) \geq f^*(b)\}} ~ dx & = & \int_0^a \int_0^{+\infty} \1_{\{ y: f(x) 1_{\{f(x) \geq f^*(b)\}} \geq y \}}(t) ~dt ~dx \\ 
& = & \int_0^a \int_0^{+\infty} \1_{\{ y: f(x) \geq \max(y,f^*(b))\}}(t) ~dt ~dx \\
& = &  \int_0^{+\infty} \int_0^a \1_{\{ y: f(x) \geq \max(y,f^*(b))\}}(t) ~dx ~dt  \\
& = & \int_0^a \int_0^{+\infty} \mu (\{ x: f(x) \geq \max(t,f^*(b))\})  ~dx ~dt \\ 
& = & \int_0^a \int_0^{+\infty} \mu (\{ x: f^*(x) \geq \max(t,f^*(b))\})  ~dx ~dt\\ 
\end{eqnarray*} where we used Proposition \ref{prop:volume} in the last step.

Continuing where we left off,
\begin{eqnarray*}
\int_0^a f(x) \1_{\{f(x) \geq f^*(b)\}} ~ dx & = &
\int_0^{+\infty} \int_0^a  \mu (\{x :  f^*(x) \geq \max(t,f^*(b))\})  ~dx ~dt \\ 
& = &  \int_0^{+\infty} \int_0^a \1_{\{ y: f^*(x) \geq \max(y,f^*(b))\}}(t)  ~dx ~dt \\
& = & \int_0^a \int_0^{+\infty} \1_{\{y:  f^*(x) \geq \max(y,f^*(b))\}}(t) ~dt ~dx \\
& = & \int_0^b \int_0^{+\infty} \1_{\{ y: f^*(x) \geq y\} }(t) ~ dt ~ dx \\ 
& = & \int_0^b f^*(x) ~ dx,
\end{eqnarray*} where the third equality used that $f^*(x)$ is strictly decreasing on the right at $b$ to replace the inequality $f^*(x) \geq f^*(b)$ with $x \leq b$.  This concludes the proof for the case when $f^*(x)$ is strictly decreasing to the right.

Now suppose $f^*(x)$ is strictly decreasing from the left at $x=b$. The corresponding identity is established with a nearly identical argument: 
\begin{eqnarray*} 
\int_0^a f(x) \1_{\{f(x) > f^*(b)\}} ~ dx & = & \int_0^a \int_0^{+\infty} \1_{\{ y: f(x) 1_{\{f(x) > f^*(b)\}} > y \}}(t) ~dt ~dx \\ 
& = & \int_0^a \int_0^{+\infty} \1_{\{ y: f(x) > \max(y,f^*(b))\}}(t) ~dt ~dx \\
& = &  \int_0^{+\infty} \int_0^a \1_{\{ y: f(x) > \max(y,f^*(b))\}}(t) ~dx ~dt  \\
& = & \int_0^a \int_0^{+\infty} \mu (\{ x: f(x) > \max(t,f^*(b))\})  ~dx ~dt \\ 
& = & \int_0^a \int_0^{+\infty} \mu (\{ x: f^*(x) > \max(t,f^*(b))\})  ~dx ~dt \\ 
& = &  \int_0^{+\infty} \int_0^a \1_{\{ y: f^*(x) > \max(y,f^*(b))\}}(t)  ~dx ~dt \\
& = & \int_0^a \int_0^{+\infty} \1_{\{y:  f^*(x) > \max(y,f^*(b))\}}(t) ~dt ~dx \\
& = & \int_0^b \int_0^{+\infty} \1_{\{ y: f^*(x) >  y\} }(t) ~ dt ~ dx \\ 
& = & \int_0^b f^*(x) ~ dx,
\end{eqnarray*} where we used that $f^*(x)$ is strictly decreasing from the left at $x=b$ to replace the inequality $f^*(x) > f^*(b)$ with $x < b$. 
\end{proof}

\begin{proof}[Proof of Proposition \ref{prop:cint2}] That the second term under the integral integrates to $\delta f^*(b) \mu(S)$ is clear. We only sketch the rest of the proof, as it is almost identical to the second part of Proposition \ref{prop:integral}. Indeed, we proceed exactly as in Proposition \ref{prop:integral} until the penultimate line, where the inequality $f^*(t) > f^*(b)$ is now replaced with $t \leq b^l$ or $t<b^l$.
\end{proof}

\end{document}